\documentclass[aps,pra,11pt,onecolumn,twoside,nobalancelastpage,
superscriptaddress,
floatfix
 ]{revtex4}

\usepackage{color}
\usepackage{bm}
\usepackage{amsfonts}
\usepackage{amsmath}
\usepackage{amssymb}
\usepackage{epsfig}
\usepackage{mathrsfs}
\usepackage{verbatim}
\usepackage{latexsym}
\usepackage{subfig}

\newcommand{\proof}{{\it Proof: }}
\newcommand{\proofend}{$\Box$}

\newcommand{\ket}[1]{|#1\rangle}
\newcommand{\bra}[1]{\langle#1|}

\newcommand{\proj}[1]{\ket{#1}\!\bra{#1}}

\newcommand{\be}[0]{\begin{equation}}
\newcommand{\ee}[0]{\end{equation}}
\newcommand{\bea}[0]{\begin{eqnarray}}
\newcommand{\eea}[0]{\end{eqnarray}}

\newcommand{\bei}{\begin{itemize}}
\newcommand{\eei}{\end{itemize}}

\def\>{\rangle}
\def\<{\langle}

\newcommand{\pacc}{p_{\text{A}}} %acceptance in 0-communication
\newcommand{\psucc}{p_{\text{S}}} %success in 0-communication
\newcommand{\pdist}{p^{\mu}_{\text{S}}} %success in distributional
\newcommand{\protocol}{\Pi}

\newcommand{\idop}{\mathbb{I}}
\newcommand{\tr}[1]{\text{tr}\left[#1\right]}

\newcommand{\nsbox}{B}

\newtheorem{lemma}{Lemma}
\newtheorem{propo}{Proposition}
\newtheorem{definition}{Definition}
\newtheorem{obs}{Observation}

\newlength{\dinwidth}
\newlength{\dinmargin}
\DeclareMathAlphabet{\scr}{U}{rsfs}{m}{n}

\begin{document}
\title{Quantum communication complexity advantage implies violation of a Bell inequality}

\author{H. Buhrman}
%\email{lczekaj@mif.pg.gda.pl}
\affiliation{Centrum Wiskunde \& Informatica (CWI), Amsterdam, The Netherlands}

\author{L. Czekaj}
%\email{lczekaj@mif.pg.gda.pl}
\affiliation{Faculty of Mathematics, Physics and Informatics, Gda\'nsk University, 80-952 Gda\'nsk,Poland}

\author{A. Grudka}
\affiliation{Faculty of Physics, Adam Mickiewicz University,
  Umultowska 85, 61-614 Pozna\'{n}, Poland}

\author{M. Horodecki}
\affiliation{Faculty of Mathematics, Physics and Informatics, Gda\'nsk University, 80-952 Gda\'nsk,Poland}

\author{P. Horodecki}
\affiliation{Faculty of Applied Physics and Mathematics, Gda{\'n}sk
  University of Technology, 80-952 Gda{\'n}sk, Poland}

\author{M. Markiewicz}
\affiliation{Faculty of Mathematics, Physics and Informatics, Gda\'nsk University, 80-952 Gda\'nsk,Poland}

\author{F. Speelman}
\affiliation{Centrum Wiskunde \& Informatica (CWI), Amsterdam, The Netherlands}

\author{S. Strelchuk}

\affiliation{Department of Applied Mathematics and Theoretical Physics, University of Cambridge}

\begin{abstract}
 We obtain a general connection between a quantum advantage in communication complexity and non-locality. We show that given any protocol offering a (sufficiently large) quantum advantage in communication complexity, there exists a way of obtaining measurement statistics which violate some Bell inequality. Our main tool is port-based teleportation. If the gap between quantum and classical communication complexity can grow arbitrarily large, the ratio of the quantum value to the classical value of the Bell quantity becomes unbounded with the increase in the number of inputs and outputs.
\end{abstract}

\maketitle
\section{Introduction}

The key element which distinguishes classical from quantum information theory is quantum correlations. 
The first attempt to quantify their strength 
%of the was first recognized in the EPR paradox~\cite{einstein_can_1935} and then 
was quantitatively expressed in Bell's 
theorem~\cite{bell_problem_1966}. 
They are similar to classical correlations in that one cannot take advantage of them to perform superluminal communication, yet, every attempt to explain such correlations from the point of view of classical theory -- namely, to find a local hidden variable model -- is impossible. For a long time the existence of quantum correlations was merely of interest to philosophically minded physicists, and was considered an exotic peculiarity, rather than a useful resource for practical problems in physics or computer science. This has changed dramatically in recent years -- it became apparent that quantum correlations can be used as a resource for a number of distributed information processing 
tasks~\cite{yao_class_model,yao_quant_model,burm_cleve} producing surprising results~\cite{Raz99,Buhrman_BIG}.

One area where using quantum correlations has wide-reaching practical implications is communication complexity.
% whose problems have a close relation to a number of important areas of modern physics, computer science and engineering~\cite{loadsOfCitations}. 
A typical instance of a communication complexity problem features two parties, Alice and Bob, who are given binary inputs $x$ and $y$.
%$x\in X=\{0,1\}^n$ and $y\in Y=\{0,1\}^n$ respectively, 
%distributed according to some a priori distribution $\mu$. 
They wish to compute the value of 
$f(x,y)$
%$f: X\times Y\to \{0,1\}$ 
by exchanging messages between each other. The minimum amount of communication required to accomplish the task by exchanging classical bits (with bounded probability of success) is called classical communication complexity, denoted as $C(f)$.

There are two ways to account for the communication complexity of computing a function when we want to make use of quantum correlations. In the first one, Alice and Bob share any number of instances of the maximally entangled state $|\Psi^-\>_{AB} = \frac{1}{\sqrt{2}}(|01\> - |10\>)_{AB}$ beforehand and are allowed to exchange classical bits in order to solve the problem. Another approach is to have no pre-shared entanglement, but instead allow Alice and Bob to exchange qubits. The latter type of protocol can always be converted to the former with pre-shared entanglement and classical communication. We denote the quantum communication complexity of computing the 
%function $f(x,y)$ with $p_{x,y}\ge1/2+\epsilon$, for some fixed $\epsilon > 0$  as $Q(f)$.
function $f(x,y)$ (with bounded probability of success) by $Q(f)$.

For a large number of problems, the respective quantum communication complexity is much lower compared to its classical counterpart~\cite{burm_cleve,quantum_entanglement_complexity_reduction}. 
In such cases, we say that there exists a {\it quantum advantage for communication complexity}. In other words, one achieves a quantum advantage if the quantum communication complexity of the function is lower than its corresponding classical communication complexity.

One of the most striking example of quantum advantage is the famous Raz problem~\cite{Raz99,RegevKartlag} where quantum communication complexity is exponentially smaller than classical. Another example is the ``hidden matching" problem for which the quantum advantage leads to one of the strongest possible violations of the Bell inequality~\cite{nonloc_hidden_matching}. The latter inequality plays an important role in detecting quantum correlations and certifying the genuinely quantum nature of resources at hand. Previously, to obtain an unbounded violation of a particular Bell inequality one resorted to problems with the exponential quantum advantage. 
[Here, we show that one can achieve the same result using only polynomial quantum advantage.]

As a matter of fact, the very first protocols offering quantum advantage were based on a quantum violation of certain Bell inequalities 
\cite{Buhrman_BIG}. It was even shown that for a very large class of multiparty Bell inequalities, correlations which violate them lead 
to a quantum advantage (perhaps, for a peculiar function) \cite{Brukner}. 
This indicates that non-locality often leads to a quantum advantage. 
However, there are more and more communication protocols which offer a quantum advantage, but, nevertheless, they are not known to violate any Bell inequality.

It has long been suspected~\cite{Buhrman_BIG} that quantum communication complexity and non-locality are the two sides of the same coin.
 While it is possible to convert an non-locality experiment to the communication complexity instance, the reverse has been known only for some particular examples. The question is whether this relationship holds in general, namely:

{\bf Q: }{\it Is quantum communication inherently equivalent to non-locality when solving communication complexity problems?}

%Does quantum advantage in communication complexity inherently rely on the non-local nature of correlations which are utilized by the parties?} 

Until now, there were only two concrete examples where one could certify quantum correlations in the context of communication complexity by providing a quantum state and a set of measurements whose statistics violate some Bell inequality. The first case is the ``hidden matching" problem and the second one is a theorem, which states that a special subset of protocols that provide quantum advantage also imply the violation of local realism~\cite{Buhrman_BIG}. To get the violation of Bell inequalities obtained from the examples above, one had to perform an involved analysis which relied on a problem-specific set of symmetries. Thus, such an approach cannot be generalized to an arbitrary protocol for achieving a quantum advantage in the communication complexity problem. 

In this paper, we show that given {\it any} (sufficiently large) quantum advantage
in communication complexity, there exists a way of obtaining measurement statistics which violate some {\it linear} Bell inequality. This completely resolves the question about the equivalence between the quantum communication and non-locality: whenever a protocol computes the value of the function $f(x,y)$ better than the best classical protocol, even with a gap that is only polynomial, then there must exist a Bell inequality which is violated. 

%and achieves $(Q(f))^4$, then 
% which in turn implies that some non-local quantum correlations. 

We provide a universal method which takes a protocol which achieves the quantum advantage in any single- or multi-round communication complexity problem and uses it to derive the violation of some linear Bell inequality. This method can be generalized to a setting with more than two parties. 
%Moreover, our method gives rise to the strongest nonlocality test using linear Bell inequalities with {\it binary} observables. The latter makes it more amenable for  experimental verification. 
Our Bell inequalities lead to a so-called unbounded violation (see~\cite{junge_2_part_extr}): the ratio of the quantum value to the classical value of the Bell quantity can grow arbitrarily large with the increase of the number of inputs and outputs, whenever $(Q(f))^4<C(f)$. In particular, an exponential advantage leads to the exponential ratio. 

Our method consists of two parts. In the first part, given a protocol which computes a function $f$ by using $Q(f)$ qubits, 
and the optimal classical error probability achievable with $(Q(f))^4$ bits, we construct the corresponding linear Bell inequality. In the second part, we use the quantum protocol to construct a set of quantum measurements on a maximally entangled state which leads to the violation of the Bell inequality above. The central ingredient of our construction is the recently-discovered port-based 
teleportation~\cite{portbasedtelep1,portbasedtelep2}. 

For one-way communication complexity problems we develop a much simpler method which is based on the remote state preparation and results in a non-linear Bell inequality.

%\tred{Interestingly, there have been many results in opposite direction, namely if one wants to turn a Bell inequality into 
%a communication complexity protocol, so that violation of Bell inequality implies quantum advantage. Indeed,
%even the original quantum advantage has been based on violation of Bell inequalities  - in the form of GHZ paradox 
%or just CHSH inequality \cite{Buhrman_BIG}. Quite general result was provided in \cite{}. }

%The first part of our method is obtained using techniques similar to the ones in the proof of Theorem 11 in~\cite{francuzi_eff_bound2}, where upper bound for quantum communication complexity in case of inefficient detectors \cite{mas02} is proposed. By a slight modification of the proof of Theorem 11 we derive a Bell inequality instead of upper bound.
%The proof of implication (ii) relies, among other things, on the remote state preparation protocol~\cite{remote_state_prep}.

%Combining the above we obtain that whenever there is a quantum advantage -- a quantum protocol requires fewer qubits to compute the value of the function than what is prescribed by the classical bound -- then there exits a Bell inequality which is violated by the statistics  of the constructed measurements.
%
%Our construction is summarized on FIG.~\ref{fig:bounds}.

\section{Main Part}

%More precisely, consider a function $f(x,y)$, such that Alice has $x$ and Bob has $y$. We assume some a priori
%distribution $\mu$ over $x,y$, and we want to determine its one-way communication complexity -- the minimal number of bits which have to be sent from Alice to Bob in order to compute the value of $f$.

\subsection{Quantum communication complexity protocol}

We start by defining a general quantum multi-round communication protocol. Two parties, Alice and Bob receive inputs $x\in X=\{0,1\}^n$ and $y\in Y=\{0,1\}^n$ according to some distribution $\mu$ and their goal is to compute the function $f: X\times Y\to \{0,1\}$  by exchanging qubits over multiple rounds. We will further use subscripts for the system names to denote the round number. The parties proceed as follows.
%The $i$-th round consists of the following steps: 
\begin{enumerate}

\item Alice, applies $U_x^{A_0\to M_1A_1}$ on her local state $\rho_{A_0}$ and sends $\rho_{M_1}$ to Bob. In general, $M_1$ may be entangled with $A_1$, which remains with Alice.

\item Bob performs $U_y^{M_1 B_0\to M_2 B_1}$ on the state $\rho_{M_1}\otimes\sigma_{B_0}$. Then he sends back the system ${M}_2$ to Alice, keeping $B_1$.

\item Parties repeat steps $1$ and $2$ for $r-1$ rounds. In the last round, instead of communicating back to Alice, Bob measures the observable $o_y$ and outputs the value of the function $f$. The observable $o_y$ acts on the system $M_{2r-1}$ and Bob's memory $B_{r-1}$.
\end{enumerate}

%
%First, Alice, applies $U^x_{M_1A_1}$ according to her input $x$ and sends $M_1$ to Bob.  In general, $M_1$ may be entangled with part of the state which remains with Alice, $A_1$. 
%
%Then, Bob performs some unitary $U^y_{M_1 B_1}$ on the $M_1$ which depends an his input $y$ and some of qubits in his possession denoted as $B_1$. Then he sends back system $\widetilde{M}_1$ to Alice. 
%
%In the next round, Alice performs $U^x_{\widetilde{M}_1 A_1}$ and sends $M_2$ to Bob, remaining with $A_2$.

%The exchange continues for $r$ rounds and at the end, Bob measures the observable $O_y$ which acts on the system $M_r$ and Bob's memory $B_{r-1}$ and outputs the value of the function $f$.

The above protocol may be transformed to the form where a one-qubit system is exchanged between Alice and Bob at any round. To achieve this, we split the $q$-qubit message from Alice to Bob (or vice versa) into $q$ rounds of one-qubit transmission and modify the protocol as follows. We start from the initial state which has the form:
\begin{equation}\label{initialstate}
\ket{\rho_A^{M}}\ket{\theta^C_A}\ket{\sigma_B^{M}},
\end{equation}
where $\ket{\rho_A^{M}}$ and $\ket{\sigma_B^{M}}$ describe the memory registers which belong to Alice and Bob respectively. The state $\ket{\theta^C_A}$, initially in state $\ket{\theta}=\ket{0}$ with Alice, is a one-qubit system which is used for message passing from Alice to Bob and vice-versa. In each round, Alice applies $U_{x}^i$ to $\rho\otimes\theta$, and Bob applies $U_{y}^i$ to $\sigma\otimes\theta$. In the last round, instead of applying a unitary transformation, Bob performs a measurement. One may view unitaries $U_{x}^i$ and $U_y^i$ as controlled gates acting on the memory with the one-qubit register acting as a control. This implies that for given $x$, in round $i$ the state of Alice memory is spanned on at most $2^i$ orthogonal vectors. This observation will be crucial for the construction of a memoryless quantum protocol. Thus, we can transform any given protocol which requires $Q$ qubits of communication into one which makes use of $2Q$ one-qubit exchanges.

\subsection{From the protocol with memory to the memoryless protocol}
 One shortcoming of the above protocols is that the parties are required to store the memory which may in general be entangled with the message and thus restrict the range of possible operations on either side. We get rid of this requirement by converting the above protocol with memory to the memoryless one. For the memoryless protocol, both parties compress their local memory and send it along with the messages.
%  As before, parties start with the initial state~\eqref{initialstate}.
%
%\begin{enumerate}
%
%\item Alice, applies $U_x^{A\to M_1A_1}$ on her local state $\rho_{A}$ and sends $\rho_{M_1}$ and $A_1$ to Bob. 
%
%\item Bob performs $U_y^{M_1 B\to  M_2 B_1}\otimes\mathbb{I}^{A_1\to\A_1}$ on the $\rho_{M_1A_1}\otimes\sigma_{B}$. Then, he sends Alice the subsystems ${M}_1$, $B_1$ and $A_1$.
%
%\item Parties repeat steps $1$ and $2$ for $r-1$ rounds but in the last round, instead of communicating back to Alice, Bob measures the observable $o_y$ which acts on the system $M_r$ and Bob's memory $B_{r-1}$ and outputs the value of the function $f$.
%\end{enumerate}

The following lemma, which is a consequence of Yao's Compression Lemma~\cite{yao_quant_model,Kremer1995} shows that, rather surprisingly, sending memory alongside the message does not impact communication complexity by much.

\begin{propo}
For any $Q$-qubit quantum communication protocol (without prior entanglement) there exists a $Q^2 + 2 Q$-qubit quantum communication protocol such that Alice and Bob do not need any local quantum memory
that persists between the rounds.
\end{propo}

\subsection{Quantum measurements from the memoryless quantum communication complexity protocol} 

We now show how to convert a multi-round protocol for computing $f(x,y)$ which gives a quantum advantage to the violation of a linear Bell inequality. The method for converting one-way protocols was introduced in~\cite{previous_paper}. 
It relies on remote state preparation~\cite{remote_state_prep}, and therefore is not extendable to a multi-round protocol. Another downside of the latter method was that it produced a nonlinear Bell inequality, whereas our method gives rise to a linear Bell inequality. Our protocol is based on 
the recently introduced method of port-based teleportation which we briefly review below.

{\bf Port-based teleportation}. In deterministic port-based teleportation, the two parties share $N$ pairs of maximally entangled qudits $|\Psi^-\>_{A_1B_1}\otimes\dotsb\otimes |\Psi^-\>_{A_NB_N}$, each of which is called a 'port'. To transmit the state $|\Psi^{in}\>_{A_0}$, the sender performs the square-root teleportation measurement given by a set of POVM elements $\{\Pi\}_{i=1}^{N}$ 
(precisely defined in Eqn. (27) of~\cite{portbasedtelep2}) on all the systems $A_i$, $i=0,...,N$, obtaining the result $z=1\ldots N$. Then, he communicates $z$ to the receiver who traces out the subsystems $B_1...B_{z-1}B_{z+1}...B_N$ and remains with the teleported state $|\Psi^{out}\>_{B_z}$ in the subsystem $B_z$. Teleportation always succeeds and the fidelity of the teleported state with the original is $F(|\Psi^{in}\>_{A_0},|\Psi^{out}\>_{B_z})\geq 1 - \frac{d^2}{N}$. The cost of the classical communication from sender to receiver is equal to $c=\log_2 N$. The distinctive feature of this protocol is that unlike the original teleportation, it  does not require a correction on the receiver's side. 

{\bf Constructing quantum measurements}. 
Using port-based teleportation we can now construct the relevant quantum measurements. Parties start with the initial state~\eqref{initialstate} 
%(with the number of ports $M$ given below) 
and perform the following protocol.
\begin{enumerate}
\item Alice applies $U_x^{A_0\to M_1A_1}$ on her local state $\rho_{A_0}$. She obtains the state of size $Q_1=\log\text{dim}M_1 +\log\text{dim}A_1$  which is teleported to Bob at once using $N_1$ ports each of dimension $2^{Q_1}$. This consumes $N_1$ ports. Alice does not communicate the classical teleportation outcomes $\{i^A_1\}$, $|\{i^A_1\}| = N_1$ with $i^A_1\in\{1,\ldots,N_1\}$ to Bob.

\item Bob applies the local unitary $U_y^{M_1 B_0\to  M_2 B_1}$ to each of the ports (he does not know the value of $i_1$) and teleports each of the $N_1$ states one-by-one by applying the teleportation measurement using $N_2$ ports each of the dimension $2^{Q_2}$ where  $Q_2 = \log\text{dim} M_2 + \log\text{dim}B_1 + \log\text{dim}A_1$. This consumes $N_1 N_2$ ports. Bob keeps the set of $ N_2$ teleportation outcomes $\{i^B_{1,1},\ldots,i^B_{1, N_2}\}$, $|\{i^B_{1,1},\ldots,i^B_{1, N_2}\}|=N_1N_2$ where for each $j=1\ldots N_2$, $i^B_{1,j}\in \{1,\ldots,N_2\}$.
\item Parties repeat steps 1 and 2 for $r-1$ rounds.
\end{enumerate}

At the end of the protocol we obtain the set of measurements which map the generic communication protocol into the set of correlations:
\begin{align}\label{corr_from_proto}
 p(&\{i_1^A\},\{i^B_{1,1},\ldots,i^B_{1,N_1}\},\{i^A_{2,1},\ldots,i^A_{2,N_1 {N}_2}\},\ldots,\nonumber \\&\{ i^B_{r,1},\ldots,i^B_{r,N_1N_2\ldots N_{2r-1}}\},\{ o_1,\ldots,o_{N_1N_2\ldots N_r}\}|x,y ),
\end{align}
where $\{o_j\}$ are the final teleportation measurements in round $r$ on Bob's side. A single round of the protocol is depicted in Fig.~\ref{figure:single-round-pbt} and the entire protocol is depicted in Fig.~\ref{figure:port-based-protocol}. 

\begin{figure}[!tbp]
  \centering
  \subfloat[Single round of the protocol.]{\includegraphics[width=0.44\textwidth]{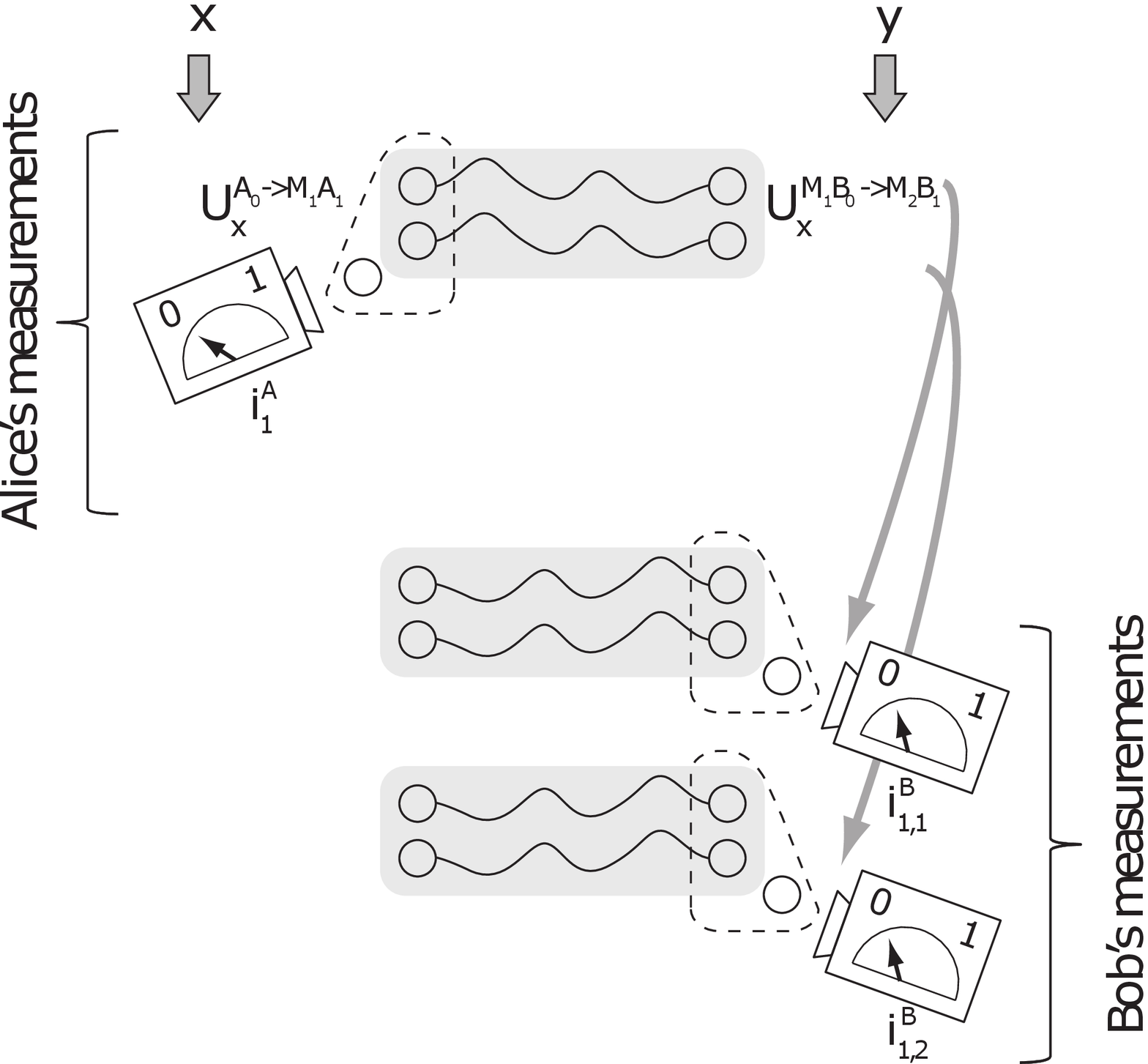}\label{figure:single-round-pbt}}
  \hfill
  \subfloat[Constructing quantum measurements. $A$ denotes Alice's local subsystem]{\includegraphics[width=0.38\textwidth]{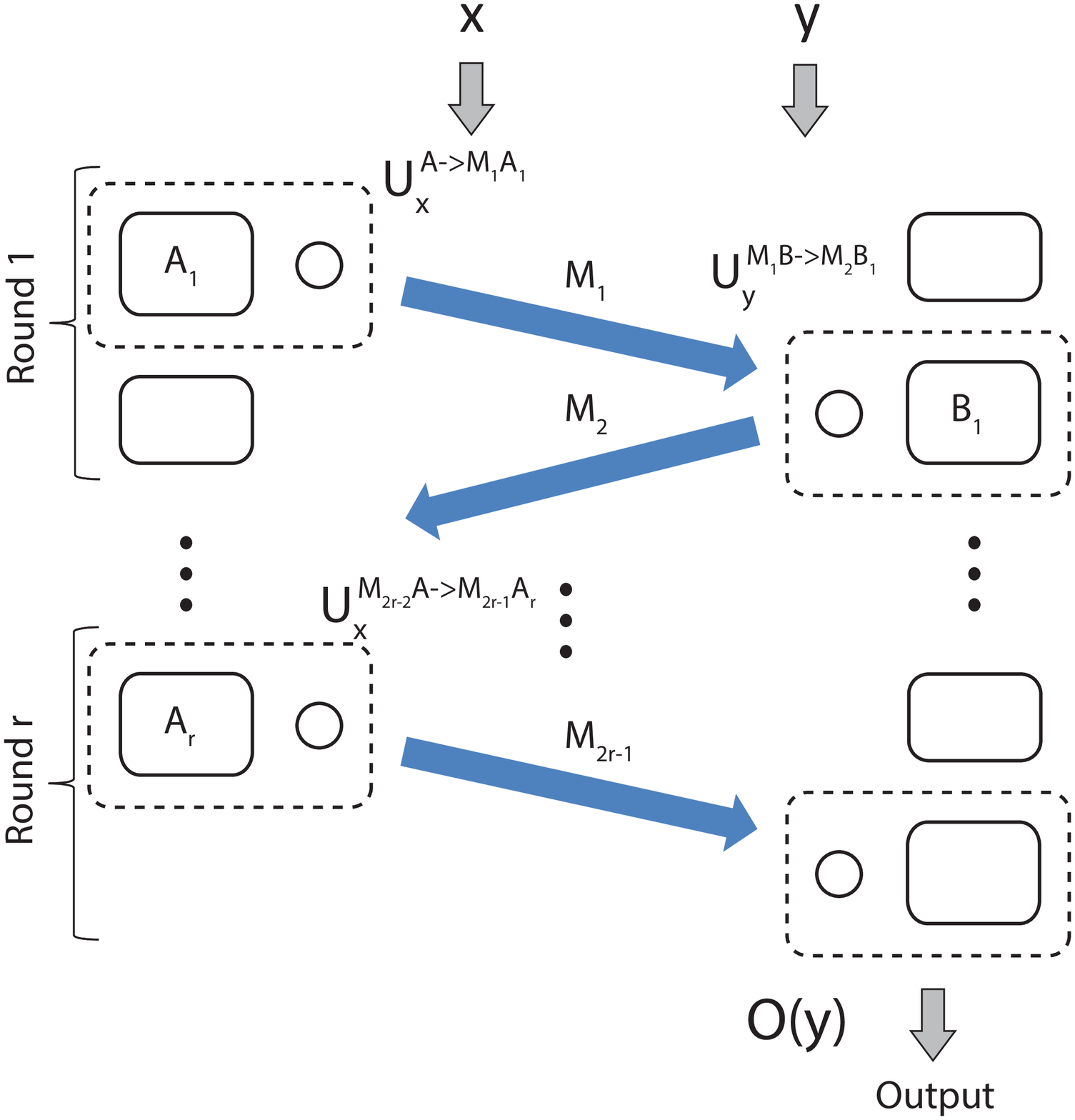}\label{figure:port-based-protocol}}
  \caption{The structure of the protocol.}
\end{figure}

%\begin{figure}
%	\centering
%		\includegraphics[scale=0.40]{correlations-to-protocol-fig2.eps}
%	\caption{Single round of the protocol.} \label{figure:single-round-pbt}
%\end{figure}
%
%\begin{figure}
%	\centering
%		\includegraphics[scale=0.35]{multiway-protocol-fig.eps}
%	\caption{Constructing quantum measurements.} \label{figure:port-based-protocol}
%\end{figure}
%

%{\it Protocol with correlations and classical communication}.
{\bf Simulating the memoryless quantum protocol}. 
The last part of the puzzle is a method of simulating the memoryless quantum protocol using the above correlations and classical communication.
\begin{lemma}
Given the memoryless protocol for computing $f$ which uses $Q$ qubits of communication and achieves the success probability $p_{succ} \geq 1/2 + \epsilon$, $\epsilon >0$, one can simulate it using correlations~\eqref{corr_from_proto} and $O(Q^2)$ bits of classical communication with the success probability $p_{succ} \geq 1/2 + (1-2^{- Q})^{2Q} \epsilon$.
\end{lemma}
\proof 
having access to correlations~\eqref{corr_from_proto}, Alice and Bob exchange their respective outcomes of the teleportation measurements which amount to $\log_2 N_1{N}_2N_3\ldots N_{2r-1}$ bits of communication. This finalizes the port-based teleportation and thus simulates the corresponding memoryless protocol. After exchange, Bob returns $o_L$ where $L$ denotes the last index which he received from Alice. 

The above protocol is equivalent to $2r$ rounds of port-based teleportation employed for the memoryless protocol. Setting $\log_2 N_i = 4 Q$ makes the fidelity of teleportation on each step to be $F\geq (1-2^{- Q})$ and the corresponding success probability $p_{succ} \geq 1/2 + F^{2r} \epsilon$ where $p_{succ} \geq 1/2 + \epsilon$ is the success probability of the original quantum protocol. Bounding the number of rounds $r$ by the total amount of quantum communication $Q$, we get $p_{succ} \geq 1/2 + (1-2^{- Q})^{2Q} \epsilon$. Thus, the total amount of classical communication is bounded above by $2Q^2$.
\proofend

\subsection{Construction of a Bell inequality and its violation}
Let us sum up the whole construction. Firstly, we start with quantum multi-round protocol to compute $f$ which uses quantum communication and no shared entanglement. This protocol requires $Q$ qubits of communication and achieves $p_{succ} \geq 1/2 + \epsilon$. In this protocol, Alice and Bob may use local quantum memory between rounds. Second, we construct the protocol without local quantum memory which increases the cost of communication to $O(Q^2)$.
The memoryless protocol is then used to obtain correlations in the form~\eqref{corr_from_proto}. These correlations together with classical communication are used to recover the original communication complexity protocol which computes $f$. This protocol uses $O(Q^4)$ bits of classical communication and achieves success probability $p_{succ} \geq 1/2 + (1-2^{- Q})^{2Q} \epsilon$. 

Now, if for a function $f(x,y)$ there exists a gap between $C(f)>(Q(f))^4$ with $p_{succ}=1/2+\delta$ for the classical communication complexity protocol, and $\delta\ll\epsilon$ -- then we observe the quantum violation of the Bell inequality of the form:

\begin{eqnarray}
\sum_{x,y}\mu(x,y)\sum_{q\in \cal P}p\left( o_{q} = f(x,y)|x,y\right) \leq 1/2+\delta,
\label{eq:mpbell}
\end{eqnarray}
where 
$\mu$ is a probability measure on $X\times Y$, the set $\cal P$ denotes the set of all paths from the root to the leaves of length $2r-1$
of the tree formed by the subsequent outputs of Alice and Bob in the protocol and $p\left( o_{q} = f(x,y)|x,y\right)$ is the marginal probability which comes from summing over all  indices which do not explicitly appear in the path $q$ (cf. Fig~\ref{figure:typicaloutputs}). With the exception of the last level, every node on the $i$-th level 
has $N_i$ children which correspond to the outcome of the $i$-th round of teleportation. The leaves of the tree correspond to the outcomes of Bob's binary observable, which is his guess of the value of the function $f(x,y)$. 
(Note that in the Bell inequality, there appear only special outputs -- those given by the paths of length $2r-1$ from the root to the leaves -- 
while in general, outputs will be given by all sequences composed by choosing one node from every level.) 
\begin{figure}
\includegraphics[width=0.38\textwidth]{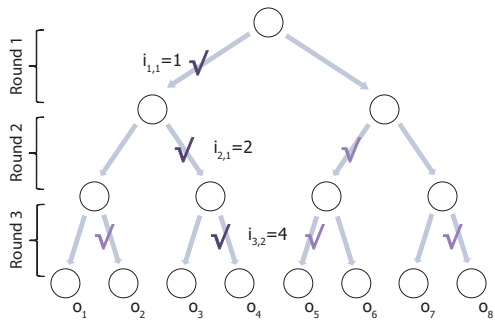}
  \caption{Selecting outputs for the Bell inequality in a $3$-round protocol: after Alice's teleportation measurement in the first round the state ended up in port $1$. Then, Bob teleports each of the two ports from the array that he used in the previous round, obtaining the outcomes $2$ and $3$ for ports $1$ and $2$ respectively. Lastly, Alice performs a teleportation measurement for each of her four ports, obtaining the outcomes $2,4,5,8$ for the ports $1,2,3,4$ respectively. A define a path $q$ to be a sequence of teleportation outcomes: $q = \{i_{1,1}=1, i_{2,1}=2, i_{3,2}=4\}$.}\label{figure:typicaloutputs}
\end{figure}
%sequences of indices exchanged by Alice and Bob (index $k$ belongs to $I$ when it is of the form  $(i_{2r-1})_{\ldots ({i}_2)_{i_1} }$).
%In other words, among all possible sequences of indices appearing in every round we select only those which have the following structure: in first round at Alice's side, we take all possible indices, and Bob's site we take only those indices which
%are from the block indicated by Alice index, and so on.
\subsection{Large violation of a Bell inequality from communication complexity}
Our results immediately imply that whenever $C(f)>(Q(f))^4$, we obtain an unbounded violation of the Bell inequality -- the ratio of the quantum to classical value of our Bell inequality
grows arbitrarily when we increase the number of inputs and outputs 
\cite{junge_2_part_extr,junge_2_part_nonmax,Buhrman_BIG,nonloc_hidden_matching,KVGame,Palazuelos_bounds}. We now introduce several definitions which enable us to contrast the performance of the quantum and classical protocols.
 
 \begin{definition}
For the arbitrary protocol $P$ computing the function $f(x,y)$ exchanging $C_{P}$ messages, denote
\be
\tilde B_P=\sum_{x,y}\mu(x,y) p(h_P(a,b)=f(x,y)|xy)
\label{eq:Bell}
\ee
 to be the Bell value achievable by some protocol $P$. The 'shifted' Bell value achievable by the protocol $P$ is $B_P = \tilde B_P-\frac{1}{2}$.
\end{definition}

The relation of the shifted Bell value with the success probability is straightforward: if a protocol $P$ computes the function with the success probability $q\ge \frac{1}{2}$, then this is equivalent to saying that it achieves the Bell value $B_P = q-\frac{1}{2}$.

We shall need the following lemma which provides the expression for the quantum to classical Bell inequality violation ratio.
\begin{lemma}\label{lemma:bellvaluelemma}
May some quantum correlations $P_q$ allow to compute the value of the function $f$ with probability of success $\frac{2}{3}$ after exchanging $C_{P_q}$ bits. Denote 
$C\left(f,\frac23\right)$ to be the number of bits required to compute $f$ using classical resources and with success probability 
$\frac23$. Then, the ratio of quantum to classical values of the Bell inequality has the form
\be\label{eqn:ratio}
\frac{B_{P_q}}{B_{P_c}} \geq c \sqrt{\frac{C\left(f,\frac{2}{3}\right)}{C_{P_q}}}
\ee
where $B_{P_q}$  and $B_{P_c}$ are the maximal quantum and classical shifted Bell values, respectively, and $c$ is an absolute constant.
\label{lemma-Bell1}
\end{lemma}
\begin{proof}
Denote $C_{P_q}$ to be the amount of quantum communication required to achieve the probability of success $p_{q}=\frac{2}{3}$.   

May a classical protocol $P_{c}$ after the exchange of $C_{P_q}$ messages achieve the probability of success $p_c = \frac{1}{2} + \delta \le \frac{2}{3}$ for solving $f$; we get $C_{P_q} =C(f,p_c)$. 
To express the ratio~\eqref{eqn:ratio} we need to find $\delta$ in terms of communication complexity. We achieve this by using the amplification argument (see Appendix~\ref{appendixA} for the proof), which boosts the success probability to $\frac{2}{3}$ at the expense of sending at most $C(f, p_{c})$ bits of communication: 
\be
C(f, p_{c})\geq \frac13 \left(p_c-\frac12\right)^2 C\left(f, \frac23\right).
\ee
Thus, we can get the expression for $\delta$ in terms of $C(f, p_{c})$ and $C\left(f,\frac{2}{3}\right)$:
\be\label{deltabound}
\delta \leq  \sqrt{\frac{3 C(f,p_c)}{C\left(f,\frac23\right)}} = \sqrt{\frac{3 C_{P_q}}{C\left(f,\frac23\right)}}.
\ee

Using the definitions of ${B_{P_q}}$ and ${B_{P_c}}$, we get

\be
\frac{B_P}{B_{c}}=\frac{1/6}{\delta}\ge\frac{1}{6\sqrt3} \sqrt{\frac{C(f,\frac23)}{C_{P_q}}}.
\ee\proofend
\end{proof}

Given the function $f$, take $Q(f, \frac{2}{3})$ (similarly $C\left(f,\frac{2}{3}\right)$) to denote the number of qubits (classical bits) required to be exchanged in order to compute $f$ with the probability of success $p\ge2/3$. To simplify the notation, when it is apparent from the context, we will further denote $Q(f,\frac{2}{3})$ and $C\left(f,\frac{2}{3}\right)$ as $Q$ and $C$ respectively.

Using our construction, it is apparent that using the quantum correlations supplemented by $Q^4$ bits of communication,
we obtain $p_{q}=\frac12 +  \frac16 (1-2^{-Q})^{2Q}$. The optimal classical protocol which uses $C(f, p_{c}) = Q^4$ bits of communication achieves some $p_{c} = \frac{1}{2} + \delta$ with $\delta\le \sqrt{\frac{3 C(f,p_c)}{C}}$ bounded 
according to Eqn.~\eqref{deltabound}. According to Lemma~\ref{lemma:bellvaluelemma}, 
\be
\frac{B_q}{B_{c}}=  \frac{\frac16 (1-2^{-Q})^{2Q}}{\delta}\ge O\left( \frac{\sqrt{C}}{Q^2}\right).
\ee

Thus, whenever the quantum communication complexity scales slower than the fourth root of the classical communication complexity, we obtain an unbounded violation of the associated Bell inequality. Let us illustrate it with a few examples.

\subsection{Examples}
{\bf Vector in subspace problem with 1-way communication}.
In this protocol, there is only one round of communication from Alice to Bob. Also, the local memory is not used. The deterministic quantum protocol requires $\log n$ qubits of communication (where $n$ is the length of the vector in the problem), while the classical communication complexity is $C\left(f,2/3\right) = \Omega(\sqrt[3]{n})$~\cite{RegevKartlag}. 

Knowing the quantum protocol $P_q$ explicitly, we obtain a stronger Bell inequality because we do not need to invoke any approximations. Using $5 \log n$ bits of communication and correlations \eqref{corr_from_proto}, we can achieve the quantum success probability of $p_{q}=1/2+1/2(1-2^{-5\log n})^{10\log n}$, while the classical protocol using the same amount of communication achieves $p_{c}=1/2+\delta$, where $\delta^2\leq \frac{5\log n}{c\sqrt[3]{n}}$, for some constant $c$. Thus, the ratio of quantum to classical values is:
\begin{equation}
\frac{B_{qu}}{B_{c}}=\frac{1/2(1-1/n)}{ \sqrt{ 5\log n / c\sqrt[3]{n}}} = \Omega\left(\frac{\sqrt[6]{n}}{\sqrt{\log n}}\right).
\end{equation}

{\bf Vector in subspace problem with 2-way communication (Raz original problem~\cite{Raz99})}. 
In this protocol, Alice sends Bob the quantum state of the size $\log n$ (where $n$ is the length of the vector in the problem) and then receives the state of the same size. As in the previous example, parties do not use local memory. There exists a deterministic quantum protocol for this problem.
The classical communication complexity is $C\left(f,2/3\right) = \Omega(\sqrt[4]{n}/\log n)$. But using only $10 \log n$ qubits of communication and correlations \eqref{corr_from_proto}, we get $p_{q}=1/2+(1-2^{-\log n})^2$. The classical protocol using the same amount of communication achieves $p_{c}=1/2+\delta$ where $\delta^2\leq c\frac{10\log^2 n}{\sqrt[4]{n}}$, for some constant $c$. Thus, the ratio of quantum to classical values is:
\begin{equation}
\frac{B_{qu}}{B_{c}}=\frac{1/2(1-1/n)^2}{ \sqrt{ c 10\log^2 n / \sqrt[4]{n}}} = \Omega\left(\frac{\sqrt[8]{n}}{\log n}\right).
\end{equation}

\subsection{One-way communication complexity problems}\label{sec:oneway}
We now detail the scenario when Alice is allowed to send a single message to Bob in order to introduce a very different approach to obtain the violation of a Bell inequality. In this case, state preparation protocol on Alice's side followed by the measurement of a quantum state by Bob will suffice. Also, there is no need for the local quantum memory on either side because one does not have to preserve the state of the communication protocol. Therefore, the role of the port-based teleportation is played by the remote state preparation. 

One marked difference of this approach is that it consumes a significantly smaller amount of entanglement. Also, in this setting, we are obtain the non-linear Bell inequality which explicitly features the probability of Bob guessing the communication from Alice -- something which is not possible using the method which relies on the port-based teleportation.

We first outline the remote state preparation protocol, and then construct the relevant Bell inequalities below.

{\bf Remote state preparation.}
In the remote state preparation, Alice and Bob share a maximally entangled qudit state $\ket{\Phi^+}_{AB} = \frac{1}{\sqrt{d}}\sum_{i=0}^{d-1}|i\rangle_A|i\rangle_B$. 
Alice wants to prepare a known quantum state $\ket{\phi}$ on Bob's side by acting only on her share of the qudit, requiring no post-processing on his side. To achieve this, she performs a measurement with elements $\{\proj{\phi^*},\idop-\proj{\phi^*}\}$, where $|\phi^*\rangle$ is a conjugation of $|\phi\rangle$ in the computational basis, on her part of $\ket{\Phi^+}_{AB}$, followed by the communication of the classical outcome to Bob if she measured $\proj{\phi^*}$ (we denote this outcome as $1$). This protocol has a very low probability of success $\frac{1}{d}$. We discuss the techniques to amplify it in the Appendix~\ref{appendixA}.

{\bf Correlations.}
Applying the remote state preparation protocol to our communication complexity problem, we obtain the following correlations:
\begin{equation}\label{eq:singlecorr}
p(a,b|x,y) = \tr{(M^a_x\otimes M^b_y)\rho_{AB}},
\end{equation}
where $\{M^a_x\}$ are the POVM elements from the remote state preparation and 
$\{M^b_y\}$ describes Bob's measurements on the shared state $\rho_{AB}$. 
In the current setup, the number of the binary observables equals of Alice and Bob is equal to the number of inputs $x$ and $y$. The correlations~\eqref{eq:singlecorr} are obtained by acting on a single instance of the entangled state whereas the multi-round approach uses in the order of $2^Q$ states.

We define the following success probabilities:
\bei
\item  $p_A$ - probability that Alice succeeded, i.e. her outcome is $1$ (averaged over all observables by the measure $\mu$)
\begin{equation}
p_A= \sum_{x,y} \mu(x,y) p(a=1|x,y)\label{eq:padef}.
\end{equation}
This probability turns out to be equal to Bob successfully `guessing' the communication from Alice in the absence of communication from the latter. 
\item $p_B$ - conditional probability, that Bob's outcome is equal to value of the function,
given that Alice succeeded
\begin{equation}
p_B= \sum_{x,y} \mu(x,y) p(b=f(x,y)|x,y,a=1)\label{eq:pbdef}.
\end{equation}
\eei

%{\bf Bell inequality without the option to abort.} 
Using roughly $m\approx 1/p_A$ instances of the state $\rho_{AB}$, Alice obtains one successful outcome $a=1$ on average. Then, Alice communicates to Bob this successful instance.

Merging $m$ instances together, we obtain following set of correlations:
\begin{equation}
 p\left(\{i\},\{ o_1,\ldots,o_{N}\}|x,y\right),
\end{equation}
where $i\in I$, $I=\{1,\ldots, m\}$ denotes the case when the remote state preparation succeeds and $\{o_i\}$ are the respective outputs. 
%Since the probability that remote state preparation faild is neglectible in the range of big problem size $n$, we get 
Thus, our Bell inequality may be written in the form~\eqref{eq:mpbell}:
\begin{eqnarray}
\sum_{x,y}\mu(x,y)\sum_{i\in I}p( i, o_{i} = f(x,y)|x,y) \leq 1/2 + \delta.
\end{eqnarray}

%{\bf Bell inequality with the option to abort.}
Now we derive a Bell inequality for the case where the parties have the option to abort at any stage of the protocol. Our inequality turns out to be nonlinear and will depend only on two parameters, $p_A$ and $p_B$.

%The way we derive inequality is described in Fig. \ref{}.
To derive the inequality, we show how Alice and Bob may guess the correct value of the function. 
In this setup, as in the previous case, Alice uses $m\approx 1/p_A$ instances of the state $\rho_{AB}$. Then Alice communicates to Bob the first instance where the outcome appeared, using $\log m \approx - \log p_A$ bits. Lastly, Bob looks at the outcome
for the successful instance, and with probability $p_B$ obtains the value of the function $f$.
%This construction is depicted in FIG.~\ref{fig:constructions}.a).

Now, if Alice and Bob share a state that admits a local-realistic description,
then the used communication cannot be smaller than the  value $C\left(p_B,n\right)$, since it is the optimal
value attainable by classical means. Thus for any local-realistic state, we must necessarily have:
\be
\log \frac{1}{p_A} \gtrsim C(p_B,n).
\label{eq:Bell1}
\ee
See Appendix~\ref{appendixB} for further details.

\subsection{Discussion}

Examples show that our protocol produces large violations which are a bit weaker than the best known ones such as 
$\frac{n}{\log^2 n}$ \cite{KVGame} or   $\frac{\sqrt{n}}{\log n}$ \cite{nonloc_hidden_matching}. This seems to be the price 
for its universality. However, it is an interesting open question, whether one can find a communication complexity protocol,
such that the obtained Bell inequality would be in some respect better than existing large Bell violations.
Another challenge is to decrease the amount of entanglement used to violate our Bell inequalities, which in our construction is exponential in the quantum communication complexity of the given problem. Similarly, the output size grows exponentially which gives rise to the question of whether there exists a more efficient method of exhibiting the non-locality of quantum communication complexity schemes. 
Finally, our method does not cover the protocols with initial entanglement. This is quite paradoxical, because protocols that use initial entanglement should be non-local even more explicitly. 
It is therefore desirable to search for a method of demonstrating the non-locality of such protocols. 

{\bf Acknowledgement}
The work is supported by EC IP QESSENCE, ERC AdG QOLAPS, the EU project SIQS, MNiSW grant IdP2011 000361, Polish National Science Centre Grant 201DEC-2011/02/A/ST2/00305 and  EC grant RAQUEL. S.S. acknowledges the support of Sidney Sussex College and FET Proactive project QALGO (600700).

\appendix
\section{Communication complexity for arbitrary success probability from
Communication complexity for fixed success probability}
\label{appendixA}

Here we shall use the amplification (pumping) argument, which is a  well known technique for increasing the success probability of a randomized protocol~\cite{book_comcomplex} by repetition, and prove the following bound for the communication required by a randomized algorithm:
\begin{equation}
C(f, \psucc)\geq \frac{1}{3}\left(\psucc-\frac12\right)^2 C\left(f, \frac23\right),
\end{equation}
where $C(f, \psucc)$ stands for communication complexity of an arbitrary (quantum or classical) randomized protocol. The bound is %valid for sufficient large $n$ and
valid for $\frac12<\psucc<\frac23$.
We use the pumping argument to show that a smaller $C(f, \psucc)$ would enable one to construct a protocol which uses less communication than $C\left(f, \frac23\right)$ and achieves $\psucc=\frac23$, and hence leads to a contradiction.

Let the protocol $\protocol$ use $C\left(\frac12+\epsilon\right)$ bits of communication to achieve $\psucc=\frac12+\epsilon$. Let us consider protocol $\protocol'$ in which Alice and Bob repeat protocol $\protocol$ $l$ times and then Bob returns as an answer the most common output of $\protocol$.
Since we are restricted to Boolean functions, the success probability $\psucc'$ of $\protocol'$ is equal to the probability that protocol $\protocol$ gives the correct answer no less than $\left\lceil l/2 \right\rceil+1$ times.
By the Chernoff bound we get:
\begin{equation}
\psucc'\geq 1 - \exp\left( - \frac{1}{2} l \epsilon^2 \right).
\end{equation}
Since we require that $\psucc'\geq\frac23$, we get that
\begin{equation}
l \geq 3 / \epsilon^2\label{eq:lepsrel}.
\end{equation}
%for some constant $\tilde{c}$.
From the communication complexity bound, it is known that in order to achieve $\psucc=\frac23$, the protocol $\protocol'$ requires at least $C\left(f, \frac23\right)$ bits of communication. On the other hand, protocol $\protocol'$ repeats protocol $\protocol$ $l$times and uses $l C(f, \frac12+\epsilon)$ bits of communication. Putting this together, we have:
\begin{equation}
l C\left(f, \frac12+\epsilon\right) \geq C\left(f, \frac23\right).
\end{equation}
Using relation~\eqref{eq:lepsrel} we get finally:
\begin{equation}
C\left(f, \frac12+\epsilon\right) \geq \frac{\epsilon^2}{3} C\left(f, \frac23\right)\label{eq:pumpingarg}.
\end{equation}
For $\frac12+\epsilon=\frac23$ our estimation leads to a communication complexity bound of $1/108 \;C\left(f, \frac23\right)$ which is much below the true value. This discrepancy comes from the non-optimality of the pumping protocol.

\section{Rigorous derivation of the Bell inequality and its violation}
\label{appendixB}

We now derive the central result of our paper - the following Bell inequality:
\begin{equation}\label{eq:maineq}
\left\lceil\log 1/\pacc + \log\log 1/\delta \right\rceil+1 \geq C_\mu(f,n,(1-\delta)p_A+\delta/2),
\end{equation}
where the classical communication complexity $C(f,p)$ is additionally parametrized by $\mu$ and the size of the problem $n$.
%and its violation stem from quantum advantage in communicatin complexity.
First, we construct a one-way protocol with classical communication which makes use of shared shared entanglement given the set of correlations. We restrict ourselves to the family of correlations $p(a,b | x,y)$ with $x,y\in\{0,1\}^n$ , $a,b\in\{0,1\}$. As usual, $a=1$ is interpreted as the success on Alice's side. When the latter occurs, we expect $b$ to hold the value of the function: $b=f(x,y)$. This restriction does not limit the generality since we may always take negation of $a,b,x,y$ which is a local operation.

We show that for any correlation $p(a,b|x,y)$, characterised by $n$, $p_A$ and $p_B$ (defined in Section~\ref{sec:oneway}), leads to the protocol $\protocol_B$ solving a problem of size $n$ using $\left\lceil\log 1/\pacc + \log\log \delta\right\rceil + 1$ bits of communication and achieving $\pdist=(1-\delta)p_B + \delta/2$ for the initial probability distribution $\mu(x,y)$.

Protocol $\protocol_B$ works as follows.
Let Alice and Bob share $\left\lceil k / p_A \right\rceil$ copies of the correlations. They use their inputs $x$, $y$ to select the appropriate measurements. Alice sends to Bob the index $i$ of the first correlation where she obtained $a=1$. Then, Bob estimates $b$ for the respective correlation $i$ and returns it as an output of protocol $\protocol_\nsbox$.
In the case when none of the boxes returned $a=1$, Alice outputs $\text{ABORT}$ and Bob  returns a random bit.

The protocol requires $\left\lceil\log k /p_A\right\rceil$ bits of communication to encode $i$ of the box and $1$ extra bit to encode the message $\text{ABORT}$.
The probability that Alice gets $a=1$ for at least one instance is $1-(1-p_A)^{k / p_A}\geq 1 - 2^{-k}$. For this case Bob returns $b=f(x,y)$ with the probability $p_B$. If Bob receives $\text{ABORT}$, he returns the proper value with probability $1/2$.
Putting $\delta=2^{-k}$ we get an overall success probability of $\pdist=(1-\delta)p_B+\delta/2$ with communication of $\left\lceil\log 1/p_A + \log\log 1/\delta \right\rceil+1$ bits.

For all the correlations with the local hidden variable model we get:
\begin{equation}
\left\lceil\log 1/\pacc + \log\log 1/\delta \right\rceil+1 \geq C_\mu(f,n,(1-\delta)p_B+\delta/2) \label{bell_distexact}.
\end{equation}

In the case when the communication complexity is given only for the fixed probability of success $\psucc=\frac23$, by the pumping argument and the fact that $C_\mu(\frac23,n)\leq C_\mu(\psucc,n)$ for $\psucc\geq \frac23$ we obtain
\begin{eqnarray}
& \left\lceil\log 1/\pacc + \log\log 1/\delta \right\rceil+1 \geq&\\
& \begin{cases}
& \frac{1}{3} \left((1-\delta)p_B+\delta/2 -\frac12\right)^2 C_\mu(\frac23,n),  \text{ if } (1-\delta)p_B+\delta/2 \leq\frac23, \\
& C_\mu(\frac23,n),  \text{ if } (1-\delta)p_B+\delta/2 > \frac23.
 \end{cases}
&\nonumber
\end{eqnarray}

Using the fact, that
correlations obtained from a quantum protocol with communication complexity $Q$ and success probability $\psucc$ are characterized by $\pacc = 2^{-Q}$ and $p_B = \psucc$ and inserting them into~\eqref{bell_distexact},
%applying pumping argument (see Appendix B),
we make the following observation:
\begin{obs}
Let $C_\mu(f,n,\psucc^C)$ be defined as in Eqn.~\eqref{eq:maineq}. If correlations
obtained by construction~\eqref{corr_from_proto} from a quantum protocol with success probability $\psucc$ and
communication complexity $Q$ do not violate the Bell inequality~\eqref{bell_distexact}, then:
\begin{equation}
Q(f,n,\psucc) \geq\label{eq:bellineqrigcc}
\max_\delta\left( C_\mu(f,n,(1-\delta)\psucc+\delta/2) - \log\log 1/\delta\right) - 2.
\end{equation}
\end{obs}

To witness the violation of a Bell inequality constructed for a particular function $f$, it suffices to know how $C_\mu(f,n,\psucc^{C})$ dominates over $Q_\mu(f,n,\psucc^{Q})$ for some fixed $\psucc^{Q}\geq\psucc^{C}$~\footnote{Function $g(n)$ dominates $h(n)$ if for any constant $k$ there exists $n_0$ such that for any $n>n_0$ we have $k g(n) \geq h(n)$.}.


\begin{thebibliography}{10}
%\bibitem{my}
%L. Czekaj, A. Przysiezna, M. Horodecki, P. Horodecki, {\it Quantum metrology: Heisenberg limit with bound entanglement.} arXiv:1403.5867
%\bibitem{Advances}
%Giovannetti, V., Lloyd, S. \& Maccone, L. Advances in quantum metrology.  {\it Nat. Phot.} %\textbf{5,} 222 (2011).
\bibitem{einstein_can_1935}
A. Einstein, B Podolsky, N. Rosen, Phys. Rev. 47, 777 (1935).
\bibitem{bell_problem_1966} J. S. Bell, Physics 1, 195 (1964).
\bibitem{yao_class_model} A. C. Yao, Proc. of 11th STOC 209 (1979)
\bibitem{yao_quant_model} A. C. Yao. In Proceedings of 34th IEEE FOCS, 352
(1993).

\bibitem{burm_cleve}
R. Cleve, H. Buhrman,
%Substituting quantum entanglement for communication,
Phys. Rev. A {\bf 56}, 1201 (1997), arXiv:9704026
\bibitem{Raz99} R. Raz. In Proc. of 31st Annual ACM Symposium on the Theory of Computing, 358 (1999).
\bibitem{Buhrman_BIG}
H. Buhrman, R. Cleve, S. Massar, R. de Wolf;
%Nonlocality and communication complexity;
Rev. Mod. Phys. {\bf 82}, 665698 (2010), arXiv:0907.3584
\bibitem{quantum_entanglement_complexity_reduction}
H. Buhrman, R. Cleve, W. van Dam SIAM J.Comput. 30 (2001) 1829-1841, arXiv:9705033
\bibitem{RegevKartlag}
O. Regev, B. Klartag, Proc. of STOC '11 the 43rd annual ACM symposium on Theory of computing, 31 (2011), arXiv:1009.3640.
\bibitem{nonloc_hidden_matching}
%Nonlocal Hidden Matching Problem
H. Buhrman, G. Scarpa, R. de Wolf,
%Better Non-Local Games from Hidden Matching,
arXiv:1007.2359
%v1 [quant-ph] 14 Jul 2010
\bibitem{Brukner}
C. Brukner, M. Zukowski, J.-W. Pan, A. Zeilinger,
%Bells Inequalities and Quantum Communication Complexity;
Phys. Rev. Lett. {\bf 92}, 127901 (2004), arXiv:0210114
\bibitem{junge_2_part_extr}
M. Junge, C. Palazuelos, D. Perez-Garcia, I. Villanueva, M.M. Wolf,
Comm. Math. Phys. {\bf 300}, 715 (2010), arXiv:0910.4228
\bibitem{portbasedtelep1}
S. Ishizaka, T. Hiroshima, Asymptotic teleportation
scheme as a universal programmable quantum processor.
Physical Review Letters, 101(24):240501 (2008)
\bibitem{portbasedtelep2}
S. Ishizaka, T. Hiroshima, Quantum teleportation scheme by selecting one of multiple output ports, Phys. Rev. A 79, 042306 (2009)
\bibitem{previous_paper}
L. Czekaj, A. Grudka, M. Horodecki, P. Horodecki, M. Markiewicz,
Universal scheme for violation of local realism from quantum advantage in one-way communication complexity, arXiv:1308.5404
\bibitem{remote_state_prep}
C. H. Bennett, P. Hayden, D. W. Leung, P. W. Shor, A. Winter,
IEEE Trans. Inform. Theory, {\bf 51},56 (2005), arXiv:0307100
\bibitem{Kremer1995}
I.~Kremer. Quantum Communication. Master's thesis, 1995
\bibitem{KVGame}
H. Buhrman, O. Regev, G. Scarpa, R. de Wolf,
%Near-Optimal and Explicit Bell Inequality Vio-
lations,
IEEE Conference on Computational Complexity, 157 (2011), arXiv:1012.5043
\bibitem{Palazuelos_bounds}
C. Palazuelos,
%On the largest Bell violation attainable by a quantum state,
arXiv:1206.3695
\bibitem{junge_2_part_nonmax}
M. Junge, C. Palazuelos,
Comm. Math. Phys. {\bf 306}, 695 (2011), arXiv:1007.3043
\bibitem{book_comcomplex}
E. Kushilevitz, N. Nisan,
Communication Complexity,
Cambridge University Press 2006

\end{thebibliography}
\end{document}